\documentclass[reqno,10pt,a4paper]{amsart}

\usepackage{amssymb}
\usepackage{url}

\allowdisplaybreaks[4]

\usepackage{mathrsfs}
\let\mathcal\mathscr
\usepackage[mathscr]{eucal}

\usepackage{color}

\let\phi=\varphi
\let\kappa=\varkappa

\DeclareMathOperator{\sym}{sym}

\DeclareMathOperator{\ad}{ad}

\newcommand*{\eval}[1]{\left.#1\right|}
\newcommand*{\abs}[1]{\left|#1\right|}
\newcommand*{\Ev}{\mathbf{E}}
\newcommand*{\Ex}{\mathbf{X}}
\theoremstyle{theorem}
\newtheorem{proposition}{Proposition}
\newtheorem{corollary}{Corollary}
\newtheorem{theorem}{Theorem}
\theoremstyle{definition}

\theoremstyle{remark}

\usepackage[all]{xy}
\SelectTips{cm}{}

\usepackage{mathrsfs}
\let\mathcal\mathscr
\usepackage[mathscr]{eucal}

\newcommand{\cprime}{\/{\mathsurround=0pt$'$}}

\begin{document}
\date{\today}
\title[Coverings and nonlocal symmetries of Lax integrable
equations]{Coverings over Lax integrable equations and their nonlocal
  symmetries
} \author{H.~Baran} \address{Mathematical Institute, Silesian University in
  Opava, Na Rybn\'{\i}\v{c}ku 1, 746 01 Opava, Czech Republic}
\email{Hynek.Baran@math.slu.cz} \author{I.S.~Krasil{\cprime}shchik}
\address{Independent University of Moscow, B. Vlasevsky 11, 119002 Moscow,
  Russia} \email{josephkra@gmail.com} \author{O.I.~Morozov} \address{Faculty
  of Applied Mathematics, AGH University of Science and Technology,
  Al. Mickiewicza 30, Krak\'ow 30-059, Poland}
\email{morozov{\symbol{64}}agh.edu.pl} \author{P.~Voj{\v{c}}{\'{a}}k}
\address{Mathematical Institute, Silesian University in Opava, Na
  Rybn\'{\i}\v{c}ku 1, 746 01 Opava, Czech Republic}
\email{Petr.Vojcak@math.slu.cz}

\date{\today}

\begin{abstract}
  Using the Lax representation with non-removable parameter, we construct two
  hierarchies of nonlocal conservation laws for the 3D rdDym equation
  $u_{ty} = u_xu_{xy} - u_yu_{xx}$ and describe the algebras of nonlocal
  symmetries in the corresponding coverings.
\end{abstract}



\maketitle

\tableofcontents
\newpage

\section*{Introduction}
\label{sec:introduction}
The 3D rdDym equation~$\mathcal{E}$, \cite{Blaszak,Pavlov2003,Morozov2009}, is
an example of nonlinear integrable equations in three independent
variables. In\-te\-gra\-bi\-li\-ty here means the existence of a Lax pair with
non-removable parameter. Such equations were studied rather intensively in
recent papers, \cite{Ovsienko2010,Fer-Mos}.  In par\-ti\-cu\-lar, in our
recent papers~\cite{BKMV-2014,BKMV-2015} we gave a full de\-scrip\-ti\-on of
2D symmetry reductions for four of such equations and discussed integrability
properties of the reductions.

Using the Lax integrability of the 3D rdDym equation, we construct two
infinite hierarchies of two-component nonlocal conservation laws
(corresponding to non-negative and non-positive powers of the spectral
parameter). To these hierarchies there correspond two infinite-dimensional
Abelian coverings (in the sense of~\cite{VinKrasTrends}) which we call
\emph{positive} and \emph{negative} and denote by~$\tilde{\mathcal{E}}^+$
and~$\tilde{\mathcal{E}}^-$, respectively, and we describe the algebras of
nonlocal symmetries in these coverings.

The equation itself possesses an infinite-dimensional algebra of local
symmetries parametrized by three arbitrary functions in~$t$ and one in~$y$
plus one standing alone scaling symmetry (which allows to assign weight to the
variables under consideration), see Table~\ref{tab:loc-sym} below. We show
that all these symmetries admit lifts to both positive and negative coverings.
In addition to the extensions of local symmetries, new, purely nonlocal ones
arise in both cases.

For the covering~$\tilde{\mathcal{E}}^+$, the scaling symmetry becomes the
terminal element of the non-positive part of the Witt
algebra~$\mathfrak{W}^-$, while the $y$-dependent symmetries become a part of
the loop algebra~$\mathfrak{L}[y]$ whose coefficients are functions in~$y$; a
natural action of~$\mathfrak{W}^-$ on~$\mathfrak{L}[y]$ is defined. No new
$t$-dependent symmetry arises on~$\tilde{\mathcal{E}}^+$ and the local ones
form a graded ideal in~$\sym(\tilde{\mathcal{E}}^+)$. The exact formulation is
given in Theorem~\ref{sec:lie-algebra-struct-theo1}.

In the case of~$\tilde{\mathcal{E}}^-$ (see
Theorem~\ref{sec:lie-algebra-struct-theo-2}) the scaling symmetry becomes the
first element in the non-negative part of the Witt algebra~$\mathfrak{W}^+$
and the $t$-dependent symmetries become a part of the loop
algebra~$\mathfrak{L}[t]$. The algebra~$\mathfrak{W}^+$ acts
on~$\mathfrak{L}[t]$, while the local $y$-dependent symmetries constitute a
direct summand in~$\sym(\tilde{\mathcal{E}}^-)$.

Finally, we show that the mutually inverse recursion operators found by one of
the authors in~\cite{Morozov2012} act on~$\mathfrak{L}[y]$
and~$\mathfrak{L}[t]$ and accomplish `tunneling' between~$\mathfrak{W}^-$
and~$\mathfrak{W}^+$.

The exposition is organized as follows. In Section~\ref{sec:preliminaries}, we
introduce some basic preparatory definitions and facts needed
below. Section~\ref{sec:equation} describes the 3D rdDym equation: local
symmetries, the Lax pair, and the coverings. The main results are formulated
and proved in Section~\ref{sec:nonlocal-symmetries}.

\section{Preliminaries}
\label{sec:preliminaries}

We expose here (in a simplified, local coordinate form) the basics of the
geometrical approach to differential equations and differential coverings
following~\cite{AMS} and~\cite{VinKrasTrends}.

\subsection{Jets and equations}
\label{sec:jets-equations}

Consider~$\mathbb{R}^n$ with coordinates~$x^1,\dots,x^n$ and~$\mathbb{R}^m$
coordinated by~$u^1,\dots,u^m$. The space of \emph{$k$-jets}~$J^k(n,m)$,
$k=0,1,\dots,\infty$, carries the coordinates~$x^1,\dots,x^n$
and~$u_\sigma^j$, where~$j=1,\dots,m$ and~$\sigma$ is a symmetrical
multi-index of length~$\abs{\sigma}\leq k$,
$u_\varnothing^j=u^j$. If~$u^j=f(x^1,\dots,x^n)$ is a vector-function then the
collection
\begin{equation*}
  u_\sigma^j=\frac{\partial^{\abs{\sigma}}u^j}{\partial x^\sigma},\qquad
  j=1,\dots,m,\quad \abs{\sigma}\leq k,
\end{equation*}
is called its \emph{$k$-jet}.

At a fixed point~$\theta\in J^k(n,m)$ tangent planes to the graphs of $k$-jets
passing through this point span the \emph{Cartan plane}~$\mathcal{C}_\theta$
and the correspondence~$\mathcal{C}\colon \theta\mapsto \mathcal{C}_\theta$ is
called the \emph{Cartan distribution}. For~$k=\infty$, a basis
of~$\mathcal{C}$ consists of the vector fields
\begin{equation*}
  D_{x^i}=\frac{\partial}{\partial x^i} + \sum_{j,\sigma} u_{\sigma
    i}^j\frac{\partial}{\partial u_\sigma^j},\qquad i=1,\dots,n,
\end{equation*}
called the \emph{total derivatives}. The total derivatives commute which
amounts to the formal integrability of the Cartan distribution
on~$J^\infty(n,m)$.

Consider a submanifold in~$J^k(n,m)$ given by the relations
\begin{equation}\label{eq:5}
  F^1(x^i,u_\sigma^j)=\dots=F^r(x^i,u_\sigma^j)=0.
\end{equation}
This is a \emph{differential equation of order~$k$}. Its \emph{infinite
  prolongation}~$\mathcal{E}\subset J^\infty(n,m)$ is given by
\begin{equation*}
  D_\sigma(F^j)=0,\qquad j=1,\dots,r,\quad \abs{\sigma}\geq 0,
\end{equation*}
where~$D_\sigma=D_{x^{i_1}}\circ\dots\circ D_{x^{i_k}}$ for~$\sigma=i_1\dots
i_k$. Everywhere below we deal with infinite prolongations only and identify
them with differential equations.

The total derivatives are restrictable to infinite prolongations and these
restrictions span the Cartan distribution on~$\mathcal{E}$. Maximal integral
manifolds of this distribution are solutions.

\subsection{Symmetries}
\label{sec:symmetries}

Consider an equation~$\mathcal{E}\subset J^\infty(n,m)$. We shall assume below
that the natural
projection~$\mathcal{E}\to J^0(n,m)=\mathbb{R}^n\times\mathbb{R}^m$ is a
surjective map \emph{onto} its target\footnote{This means that the
  differential consequences of~\eqref{eq:5} do not contain $0$-order
  functions.}. Consequently, the algebra~$C^\infty(J^0(n,m))$ is embedded into
the algebra~$C^\infty{}(\mathcal{E})$.

A vector field~$X\colon C^\infty{}(\mathcal{E})\to C^\infty{}(\mathcal{E})$ is
called \emph{vertical} if~$\eval{X}_{C^\infty(J^0(n,m))}=0$, i.e., $X$ does
not contain components of the form~$\partial/\partial x^i$. A vertical
field~$X$ is a (\emph{higher}, or \emph{generalized}) \emph{symmetry}
of~$\mathcal{E}$ if it preserves the Cartan distribution,
i.e.,~$[X,\mathcal{C}]\subset \mathcal{C}$. Symmetries of~$\mathcal{E}$ form a
Lie algebra denoted by~$\sym(\mathcal{E})$.

A vector field is a symmetry if and only if it has the \emph{evolutionary}
form
\begin{equation}\label{eq:26}
  \Ev_\phi = \sum D_\sigma(\phi^j)\frac{\partial}{\partial u_\sigma^j},
\end{equation}
where summation is taken over the \emph{internal} coordinates
on~$\mathcal{E}$. Here~$\phi=(\phi^1,\dots,\phi^m)$ is a vector-function
on~$\mathcal{E}$ called the \emph{generating section} (or
\emph{characteristic}) of the symmetry. It must satisfy the equation
\begin{equation*}
  \ell_{\mathcal{E}}(\phi) = 0,
\end{equation*}
where~$\ell_{\mathcal{E}}$ is the \emph{linearization} of~$\mathcal{E}$
defined as the restriction of the operator
\begin{equation}\label{eq:27}
  \ell_F =\left\Vert \sum_\sigma\frac{\partial F^j}{\partial u_\sigma^l}D_\sigma
  \right\Vert
\end{equation}
to~$\mathcal{E}$. Generating functions form a Lie algebra with respect to the
\emph{Jacobi bracket}
  \begin{equation*}
    \{\phi,\psi\}^j=\sum\left(D_\sigma(\phi^l)\frac{\partial \psi^j}{\partial
        u_\sigma^l} - D_\sigma(\psi^l)\frac{\partial \phi^j}{\partial
        u_\sigma^l}\right),
  \end{equation*}
which in the coordinate-free way can be defined as~$\{\phi,\psi\} =
\Ev_\phi(\psi) - \Ev_\psi(\phi)$.

\subsection{Differential coverings}
\label{sec:diff-cover}

Consider the space~$\tilde{\mathcal{E}}=\mathbb{R}^s\times\mathcal{E}$,
$s\leq\infty$, and the natural
projection~$\tau\colon\tilde{\mathcal{E}}\to\mathcal{E}$. We say that~$\tau$
is an $s$-dimensional (\emph{differential}) \emph{covering} over~$\mathcal{E}$
if~$\tilde{\mathcal{E}}$ is endowed with vector
fields~$\tilde{D}_{x^1},\dots,\tilde{D}_{x^n}$ such that
\begin{equation*}
  [\tilde{D}_{x^i},\tilde{D}_{x^j}] = 0,\quad\tau_*(\tilde{D}_{x^i}) =
  D_{x^i}, \qquad i,j=1,\dots,n.
\end{equation*}
Let~$\{w^\alpha\}$ be coordinates in~$\mathbb{R}^s$ (they are called
\emph{nonlocal variables}). Then the covering structure is given by
\begin{equation*}
  \tilde{D}_{x^i} = D_{x^i} + X_i
\end{equation*}
such that
\begin{equation*}
  D_{x^i}(X_j) - D_{x^j}(X_i) + [X_i,X_j] = 0,
\end{equation*}
where
\begin{equation*}
  X_i = \sum_{\alpha} X_i^\alpha\frac{\partial}{\partial w^\alpha}
\end{equation*}
are $\tau$-vertical vector fields.

There exists a distinguished class of coverings that are associated with
two-component conservation laws of~$\mathcal{E}$. Fix two integers~$i$
and~$j$, $1\leq i<j\leq n$, and consider a differential form
\begin{equation*}
  \omega=X_i\,dx^1\wedge\dots\wedge\widehat{dx^i}\wedge\dots\wedge\,dx^n +
  X_j\,dx^1\wedge\dots\wedge\widehat{dx^j}\wedge\dots\wedge\,dx^n,
\end{equation*}
where `hat' means that the corresponding term is omitted, closed with respect
to the horizontal de~Rham differential, i.e., such that
\begin{equation*}
  D_{x^i}(X_i)=(-1)^{i+j-1}D_{x^j}(X_j).
\end{equation*}
Consider the Euclidean space~$V$ with the coordinates~$w^\sigma$, where~$\sigma$
is symmetric multi-index whose entries are any integers~$1,\dots,n$ except
for~$i$ and~$j$. Thus,~$\dim V=1$ if~$n=2$ and~$\dim V=\infty$ otherwise. Then
the system of vector fields
\begin{align*}
  \tilde{D}_{x^k} &= D_{x^k} +\sum_\sigma w^{\sigma k}\frac{\partial}{\partial
                    w^\sigma},\qquad k\neq i,j,\\
  \tilde{D}_{x^i} &= D_{x^i} + \sum_\sigma
                    \tilde{D}_\sigma(X_j)\frac{\partial}{\partial w^\sigma},\\
  \tilde{D}_{x^j} &= D_{x^j} + (-1)^{i+j-1}\sum_\sigma
                    \tilde{D}_\sigma(X_i)\frac{\partial}{\partial w^\sigma}
\end{align*}
define a covering structure
on~$\tilde{\mathcal{E}}_\omega=V\times\mathcal{E}$.  Coverings of this type
are called Abelian.

\subsection{Nonlocal symmetries}
\label{sec:nonlocal-symmetries-1}

Denote by~$\mathcal{C}$ the distribution on~$\tilde{\mathcal{E}}$ spanned by
the fields~$\tilde{D}_{x^1},\dots,\tilde{D}_{x^n}$ and let~$X$ be a field
vertical with respect to the
composition~$\tilde{\mathcal{E}} \to\mathcal{E}\to\mathbb{R}^n$. Such a field
is called a \emph{nonlocal symmetry} if it
preserves~$\tilde{\mathcal{C}}$. These symmetries form a Lie algebra denoted
by~$\sym_\tau(\mathcal{E})$. The
restriction~$\eval{X}_{C^\infty(\mathcal{E})}\colon C^\infty(\mathcal{E})\to
C^\infty(\tilde{{\mathcal{E}}})$ is called a nonlocal $\tau$-\emph{shadow}. A
nonlocal symmetry is said to be \emph{invisible} if its shadow vanishes.

In local coordinates, any~$X\in\sym_\tau(\mathcal{E})$ is of the form
\begin{equation*}
  X=\tilde{\Ev}_\phi+\sum_\alpha\psi^\alpha\frac{\partial}{\partial w^\alpha},
\end{equation*}
where~$\phi=(\phi^1,\dots,\phi^m)$, $\psi^\alpha$ are functions
on~$\tilde{\mathcal{E}}$ satisfying the equations
\begin{align*}
  &\tilde{\ell}_{\mathcal{E}}(\phi)=0,\\
  &\tilde{D}_{x^i}(\psi^\alpha)=\sum_{j,\sigma}\frac{\partial
    X_i^\alpha}{\partial u_\sigma^j}\tilde{D}_\sigma(\phi^j) +
    \sum_\beta\frac{\partial X_i^\alpha}{\partial w^\beta}\psi^\beta,
\end{align*}
where~$\tilde{\Ev}_\phi$ ans~$\tilde{\ell}_{\mathcal{E}}$ are obtained from
the expressions~\eqref{eq:26} and~\eqref{eq:27}, respectfully, by
changing~$D_{x^i}$ to~$\tilde{D}_{x^i}$. Nonlocal shadows are the
operators~$\tilde{\Ev}_\phi$ while invisible symmetries are obtained from
general ones by setting~$\phi=0$.

In particular, for coverings of the form~$\tilde{\mathcal{E}}_\omega$,
where~$\omega$ is a $2$-component conservation law, the symmetries acquire the
form
\begin{equation*}
  X=\tilde{\Ev}_\phi+\sum_\sigma D_\sigma(\psi)\frac{\partial}{\partial w^\sigma},
\end{equation*}
where~$\phi$ and~$\psi$ satisfy
\begin{align*}
  &\tilde{\ell}_{\mathcal{E}}(\phi)=0,\\
  &\tilde{D}_{x^i}(\psi)=\sum_{\sigma,k}\frac{\partial X_j}{\partial
    u_\sigma^k}\tilde{D}_\sigma(\phi^k) + \sum_\sigma\frac{\partial
    X_j}{\partial w^\sigma}\tilde{D}_\sigma(\psi),\\
  &\tilde{D}_{x^j}(\psi)=  (-1)^{i+j-1}\left(\sum_{\sigma,k}\frac{\partial
    X_i}{\partial
    u_\sigma^k}\tilde{D}_\sigma(\phi^k) + \sum_\sigma\frac{\partial
    X_i}{\partial w^\sigma}\tilde{D}_\sigma(\psi)\right).
\end{align*}

\subsection{B\"{a}cklund transformations and recursion operators}
\label{sec:bacl-transf-recurs}

Let~$\mathcal{E}_1$, $\mathcal{E}_2$ be equations. A \emph{B\"{a}cklund
  transformation} between~$\mathcal{E}_1$ and~$\mathcal{E}_2$ is the diagram
\begin{equation*}
  \xymatrix{&\tilde{\mathcal{E}}\ar[dl]_{\tau_1}\ar[dr]^{\tau_2}&\\
    \mathcal{E}_1&&\mathcal{E}_2\rlap{,}
  }
\end{equation*}
where~$\tau_1$, $\tau_2$ are coverings. When~$\mathcal{E}_1=\mathcal{E}_2$, it
is called a \emph{B\"{a}cklund auto-trans\-for\-ma\-ti\-on}. If~$\tau_1$ is
finite-dimensional and~$\gamma\subset\mathcal{E}_1$ is a graph of solution
then, generically,~$\tau_2\left(\tau_1^{-1}(\gamma)\right)$ is a
finite-dimensional manifold endowed with an integrable $n$-dimensional
distribution whose integral manifolds are solutions of~$\mathcal{E}_2$.

Consider now an equation~$\mathcal{E}$ given by~\eqref{eq:5} and the system
\begin{equation*}
  F(x^i,u_\sigma^j)=0, \qquad \ell_F(x^i,u_\sigma^j,q_\sigma^j)=0,
\end{equation*}
where~$F=(F^1,\dots,F^r)$. This system is called the \emph{tangent} equation
to~$\mathcal{E}$ and de\-no\-ted by~$\mathcal{T}\mathcal{E}$, while the
projection~$\mathrm{t}\colon\mathcal{T}\mathcal{E}\to\mathcal{E}$ is called
the \emph{tangent covering}. Sections of this covering that preserve the
Cartan distribution are identified with generating functions of symmetries.

Let~$\mathcal{R}$ be a B\"{a}cklund transformation
between~$\mathcal{T}\mathcal{E}_1$ and~$\mathcal{T}\mathcal{E}_2$. Then it
follows from the above said that it accomplishes a correspondence between
symmetries of the two equations. If~$\mathcal{E}_1=\mathcal{E}_2$ then such a
correspondence is called a \emph{recursion operator},~\cite{Mar-another}.

\section{The equation}
\label{sec:equation}

The 3D rdDym equation~$\mathcal{E}$ is of the form
\begin{equation}
  \label{eq:1}
  u_{ty} = u_xu_{xy} - u_yu_{xx}.
\end{equation}
For internal coordinates in~$\mathcal{E}$ one can choose the functions
\begin{equation*}
  u_k=u_{\underbrace{x\dots x}_{k \text{times}}},\quad
  u_{k,l}^t=u_{\underbrace{x\dots x}_{k \text{times}}\underbrace{t\dots t}_{l
      \text{times}}},\quad
  u_{k,l}^y=u_{\underbrace{x\dots x}_{k \text{times}}\underbrace{y\dots y}_{l
      \text{times}}},\qquad k\geq0,\ l\geq0.
\end{equation*}
Thus,~$u_0=u$, $u_1=u_x$, $u_{0,1}^y=u_y$, $u_{0,1}^t=u_t$, etc.  The total
derivatives acquire the form
\begin{align*}
  D_x&=\frac{\partial}{\partial x} + \sum_{k}u_{k+1}\frac{\partial}{\partial
       u_k} + \sum_{k,l}\left(
       u_{k+1,l}^y\frac{\partial}{\partial u_{k,l}^y} +
       u_{k+1,l}^t\frac{\partial}{\partial u_{k,l}^t}\right),\\
  D_y&=\frac{\partial}{\partial y} + \sum_ku_{k,1}^y\frac{\partial}{\partial
       u_k} +
       \sum_{k,l}\left(u_{k,l+1}^y\frac{\partial}{\partial u_{k,l}^y} +
       D_x^kD_t^{l-1}(u_xu_{xy} - u_yu_{xx})\frac{\partial}{\partial
       u_{k,l}^t}\right),\\
  D_t&=\frac{\partial}{\partial t} + \sum_ru_{k,1}^t\frac{\partial}{\partial
       u_k} + \sum_{k,l}\left(D_x^kD_y^{l-1}(u_xu_{xy} -
       u_yu_{xx})\frac{\partial}{\partial u_{k,l}^y} +
       u_{k,l+1}^t\frac{\partial}{\partial u_{k,l}^t}\right)
\end{align*}
in these coordinates.

\subsection{Local symmetries}
\label{sec:local-symmetries}

Local symmetries of Equation~\eqref{eq:1} are solutions to the linearized
equation
\begin{equation}
  \ell_{\mathcal{E}}(\phi)\equiv D_tD_y(\phi) - u_xD_xD_y(\phi) +
  u_yD_x^2(\phi) - u_{xy}D_x(\phi) + u_{xx}D_y(\phi) = 0.
\label{linearized_eq}
\end{equation}
The space of solutions is spanned by the functions
\begin{align*}
  \psi_0&=x u_x - 2u,\\
  \upsilon_0(B)&=Bu_y, \\
  \theta_{0}(A)&=Au_t + A'(xu_x-u) + \frac{1}{2}A''x^2\,\\
  \theta_{-1}(A)&=Au_x+A'x,\\
  \theta_{-2}(A)&=A,
\end{align*}
where~$A=A(t)$, $B=B(y)$ and `prime' denotes the $t$-derivative. To any
solution~$\phi$ there corresponds the evolutionary vector field
\begin{equation}\label{eq:7}
  \Ev_\phi=\sum_k D_x^k(\phi)\frac{\partial}{\partial u_k} +
  \sum_{k,l}\left(D_x^kD_y^l(\phi)\frac{\partial}{\partial u_{k,l}^y} +
    D_x^kD_t^l(\phi)\frac{\partial}{\partial u_{k,l}^t}\right).
\end{equation}
on~$\mathcal{E}$.

The Lie algebra structure in the space~$\sym(\mathcal{E})$ is presented in
Table~\ref{tab:loc-sym}.

\begin{table}[bct]
  \centering
  \begin{tabular}{|c|c|c|c|c|c|}
    \hline
    \rule{0em}{12pt}&$\psi_0$&$\upsilon_0(\bar{B})$&$\theta_{0}(\bar{A})$&$\theta_{-1}(\bar{A})$
    &$\theta_{-2}(\bar{A})$\\
    \hline
    \rule{0em}{12pt}$\psi_0$&$0$&$0$&$0$&
                          $-\theta_{-1}(\bar{A})$&$-2\theta_{-2}(\bar{A})$\\
\rule{0em}{12pt}$\upsilon_0(B)$&\dots&$\upsilon_0(B\,\bar{B}' - \bar{B}\,B') $&0&$0$&$0$\\

    \rule{0em}{12pt}$\theta_{0}(A)$&\dots&\dots&
                           $\theta_{0}(\bar{A}A'-A\bar{A}')$&$\theta_{-1}(\bar{A}A' -  A\bar{A}')$
    &$\theta_{-2}(\bar{A}A'-A\bar{A}')$\\

    \rule{0em}{12pt}$\theta_{-1}(A)$&\dots&\dots&\dots&$\theta_{-2}(\bar{A}A'-A\bar{A}')$&$0$\\
    \rule{0em}{12pt}$\theta_{-2}(A)$&\dots&\dots&\dots&\dots&$0$\\
    \hline
  \end{tabular}\\ \ \\
  \caption{The Lie algebra structure of~$\sym(\mathcal{E})$}
  \label{tab:loc-sym}
\end{table}

\subsection{Coverings}
\label{sec:coverings}

The 3D rdDym equation~\eqref{eq:1} possesses the linear Lax representation
\begin{equation}
  \label{eq:2}
  \begin{array}{lcl}
    w_t&=&(u_x-\lambda)w_x\\
    w_y&=&\lambda^{-1}u_yw_x,
  \end{array}
\end{equation}
where~$\lambda\neq 0$ is a non-removable parameter.  Expanding~$w$ in formal
series in~$\lambda$
\begin{equation*}
  w=\sum_{i=-\infty}^{+\infty}w_i\lambda^i.
\end{equation*}
yields, cf. \cite{Pavlov2003},
\begin{equation}
  \label{eq:3}
  \begin{array}{lcl}
    w_{i,t}&=&u_xw_{i,x} - w_{i-1,x},\\
    w_{i,y}&=&u_yw_{i+1,x}.
  \end{array}
\end{equation}
This system is infinite in both directions and thus the nonlocal
quantities~$w_i$ are not defined in a proper way.  To improve the setting,
consider two reductions of~\eqref{eq:3}: (a)~$w_i=0$ for~$i<0$ and (b)~$w_i=0$
for~$i>0$. Two hierarchies of nonlocal two-component conservation laws arise
in such a way, \cite{Pavlov2003}.  They will be called the \emph{positive} and
the \emph{negative} ones, respectively.  Our aim is to describe nonlocal
symmetries of the corresponding Abelian coverings.

Note that the positive hierarchy corresponds to the Taylor expansion of~$w$,
while the negative one is related to the Laurent expansion.

\subsubsection{The positive hierarchy}
\label{sec:positive-hierarchy}

Assume~$w_i=0$ for~$i<0$ and rewrite~\eqref{eq:3} in the form
\begin{align*}
  w_{i,t}&=\frac{u_x}{u_y}w_{i-1,y} - w_{i-1,x},\\
  w_{i,x}&=\frac{w_{i-1,y}}{u_y}.
\end{align*}
Then, due to the assumption,~$w_{0,t}=w_{0,x}=0$, or~$w_0=G(y)$ and the
defining equations of the covering are
\begin{equation*}
  \label{eq:4}
  \begin{array}{lcl}
    w_{1,t}&=&\dfrac{u_x}{u_y}G',\\[3mm]
    w_{1,x}&=&\dfrac{G'}{u_y};
  \end{array}\qquad
  \begin{array}{lcl}
    w_{i,t}&=&\dfrac{u_x}{u_y}w_{i-1,y} - w_{i-1,x},\\[3mm]
    w_{i,x}&=&\dfrac{w_{i-1,y}}{u_y},
  \end{array}
\end{equation*}
where~$i>0$ and `prime' denotes the $y$-derivative.

Without loss of generality we may assume~$G'\neq 0$ and make the change of
variables~$y\mapsto G(y)$. This transformation preserves our equation (due to
the symmetry~$\upsilon_0(B)$). Denoting the resulting nonlocal variables
by~$q_i$, $i>0$, we arrive to the covering defined by
\begin{align}
  \label{eq:5a}
  &\begin{array}{lcl}
    q_{1,t}&=&\dfrac{u_x}{u_y},\\[3mm]
    q_{1,x}&=&\dfrac{1}{u_y},
   \end{array}\\
  \intertext{and}
  \label{eq:5b}
  &\begin{array}{lcl}
    q_{i,t}&=&\dfrac{u_x}{u_y}q_{i-1,y} - q_{i-1,x},\\[3mm]
    q_{i,x}&=&\dfrac{q_{i-1,y}}{u_y}.
  \end{array}
\end{align}
Note that the quantities~$q_i$ do not form a complete set of nonlocal
variables in the covering under consideration. To have a complete collection,
let us introduce the functions~$q_i^{(j)}$ such that
\begin{equation*}
  q_i^{(0)}=q_i,\qquad q_i^{(j+1)} = \left(q_i^{(j)}\right)_y.
\end{equation*}
Then the total derivatives on the space~$\tilde{\mathcal{E}}^+$ of the covering
are given by
\begin{align*}
  \tilde{D}_x& = D_x +
               \sum_{j=0}^\infty
               \tilde{D}_y^j\left(\frac{1}{u_y}\right)\frac{\partial}{\partial
               q_1^{(j)}} + \sum_{i=2}^\infty\sum_{j=0}^\infty
               \tilde{D}_y^j
               \left(\frac{q_{i-1}^{(1)}}{u_y}\right)\frac{\partial}{\partial
               q_i^{(j)}},\\
  \tilde{D}_y&=D_y + \sum_{i=1}^\infty\sum_{j=0}^\infty
               q_i^{(j+1)}\frac{\partial}{\partial q_i^{(j)}},\\
  \tilde{D}_t&=D_t + \sum_{j=0}^\infty\tilde{D}_y^j\left(\frac{u_x}{u_y}\right)
               \frac{\partial}{\partial q_1^{(j)}} + \sum_{i=2}^\infty
               \sum_{j=0}^\infty
               \tilde{D}_y^j\left(\frac{u_x}{u_y}q_{i-1}^{(1)} -
               \tilde{D}_x\left( q_{i-1}^{(0)}\right)
               \right)\frac{\partial}{\partial q_i^{(j)}},
\end{align*}
where~$D_x$, $D_y$, $D_t$ are the total derivatives on~$\mathcal{E}$ given
above.

\subsubsection{The negative hierarchy}
\label{sec:negative-hierarchy}

We have~$w_i=0$ for~$i>0$ now. Then from~\eqref{eq:3} it follows that
\begin{equation*}
  \begin{array}{lcl}
    w_{0,x}&=&0,\\
    w_{0,y}&=&0;
  \end{array}\quad
  \begin{array}{lcl}
    w_{-1,x}&=&u_xw_{0,x,}-w_{0,t},\\
    w_{-1,y}&=&u_yw_{0,x};
  \end{array}\
  \begin{array}{lcl}
    w_{-2,x}&=&u_xw_{-1,x}-w_{-1,t}\\
    w_{-2,y}&=&u_yw_{-1,x}.
  \end{array}
\end{equation*}
Consequently,
\begin{equation*}
  w_0=\tilde{F}(t),\qquad w_{-1}=-x\tilde{F}'+G(t),\qquad w_{-2}=-\tilde{F}'u
  +\frac{1}{2}x^2\tilde{F}'' - G'x
  +H(t).
\end{equation*}
Without loss of generality we can assume~$G=H=0$. Then, after
relabeling~$r_i=w_{-i-2}$, $i=1$, $2$, \dots, we get the following defining
equations for the negative hierarchy:
\begin{equation}
  \label{eq:6}
  \begin{array}{l}
    r_{1,x}=F(u_t-u_x^2) +F'(u + xu_x) - \dfrac{1}{2}x^2F'' ,\\[2mm]
    r_{1,y}=u_y(xF'-Fu_x);
  \end{array}\quad
  \begin{array}{l}
    r_{i,x}=u_xr_{i-1,x}- r_{i-1,t},\\[2mm]
    r_{i,y}=u_yr_{i-1,x}.
  \end{array}
\end{equation}
for~$i>1$, where~$F=\tilde{F}'$. The defining equations can be simplified:

\begin{proposition}
  \label{sec:negative-hierarchy-gauge}
  There exists a gauge transformation of the space~$\tilde{\mathcal{E}}^-$
  that \textup{`}kills\textup{'} the function~$F$\textup{,} i.e.\textup{,}
  transforms~\eqref{eq:6} to
  \begin{equation}
    \label{eq:8}
    \begin{array}{l}
    r_{1,x}=u_x^2-u_t,\\[2mm]
    r_{1,y}=u_xu_y;
  \end{array}
  \quad
  \begin{array}{l}
    r_{i,x}=u_xr_{i-1,x}- r_{i-1,t},\\[2mm]
    r_{i,y}=u_yr_{i-1,x}.
  \end{array}
  \end{equation}
\end{proposition}
\begin{proof}
  Define the new nonlocal variable~$\bar{r}_1$ by
  \begin{equation}
    \label{eq:24}
    r_1=-F\bar{r}_1 - F'xu +\frac{1}{6}F''x^3.
  \end{equation}
  Substituting~\eqref{eq:24} to the left equations in~\eqref{eq:6}, we
  immediately see that
  \begin{equation*}
    \bar{r}_{1,x}=u_x^2-u_t,\qquad    \bar{r}_{1,y}=u_xu_y.
  \end{equation*}
  Let us now introduce the operator
  \begin{equation*}
    \mathcal{Y}_- = -x\frac{\partial}{\partial t} + 2u\frac{\partial}{\partial
      x}
    - 3\bar{r}_1\frac{\partial}{\partial u} + \sum_{i\geq 1}
    (i+3)\bar{r}_{i+1}\frac{\partial}{\partial \bar{r}_i}
  \end{equation*}
  and set by induction
  \begin{equation}\label{eq:25}
    r_{k}=\frac{1}{k+2}\mathcal{Y}_-(r_{k-1})
  \end{equation}
  for~$k\geq2$. Obviously,
  \begin{equation*}
    r_k=F\bar{r}_k + o(k-1),
  \end{equation*}
  where~$o(k-1)$ denotes the terms that depend
  on~$\bar{r}_1,\dots,\bar{r}_{k-1}$ only.

  Assume now that~$k>1$ and the statement is valid for the defining equations
  on~$\bar{r}_1,\dots,\bar{r}_{k-1}$. Then, substituting~\eqref{eq:25} to the
  equations on~$r_k$, we see that it transforms to
  \begin{equation*}
    F\bar{r}_{k,x}=F(u_x\bar{r}_{k-1,x}- \bar{r}_{k-1,t}),\qquad
    F\bar{r}_{k,y}=Fu_y\bar{r}_{k-1,x}
  \end{equation*}
  by the induction assumption.
\end{proof}

We forget about the `old'~$r_k$s and change the notation from~$\bar{r}_k$
to~$r_k$.

A complete set of nonlocal variables consists of the quantities~$r_i^{(j)}$
defined by
\begin{equation*}
  r_i^{(0)}=r_i,\qquad r_i^{(j+1)}=\left(r_i^{(j)}\right)_t.
\end{equation*}
The total derivatives on the covering space~$\tilde{\mathcal{E}}^-$ are of the
form
\begin{align*}
  \tilde{D}_x&=D_x + \sum_{j=0}^\infty \tilde{D}_t^j(u_x^2 - u_t)
               \frac{\partial}{\partial
               r_1^{(j)}} + \sum_{i=2}^\infty
               \sum_{j=0}^\infty \tilde{D}_t^j(u_xr_{i-1,x}- r_{i-1,t})
                              \frac{\partial}{\partial r_i^{(j)}},\\
  \tilde{D}_y&=D_y + \sum_{j=0}^\infty \tilde{D}_t^j(u_xu_y)
               \frac{\partial}{\partial r_1^{(j)}} + \sum_{i=2}^\infty
               \sum_{j=0}^\infty \tilde{D}_t^j(u_yr_{i-1,x})
               \frac{\partial}{\partial r_i^{(j)}},\\
  \tilde{D}_t&=D_t + \sum_{i=1}^\infty \sum_{j=0}^\infty r_i^{(j+1)}
               \frac{\partial}{\partial r_i^{(j)}}
\end{align*}
in these coordinates.

\subsection{Weights}
\label{sec:weights}

Let us assign the following weights
\begin{equation*}
  \abs{x} = 1,\quad \abs{u} = 2,\quad \abs{y} = \abs{t} = 0
\end{equation*}
to the dependent and independent variables. Then
\begin{equation*}
  \abs{u_k} = \abs{u_{k,l}^y} = \abs{u_{k,l}^t} = 2 - k
\end{equation*}
and to any monomial in jet variables we assign the summarized weight of its
factors.

We say that a vector field~$X$ is \emph{homogeneous} if
\begin{equation*}
  \abs{X(f)} = \abs{X} + \abs{f}
\end{equation*}
for any homogeneous function~$f$, where the integer~$\abs{X}$ depends on~$X$
only and is the weight of~$X$. All local symmetries are homogeneous in this
sense and their weight are presented in Table~\ref{tab:distr}.
\begin{table}[bct]
  \centering
  \begin{tabular}{c|ccc}
    \hline
    Weight:&$-2$&$-1$&$0$\\
    \hline
           &&&$\psi_0$\\
           &$\theta_{-2}(A)$&$\theta_{-1}(A)$&$\theta_{0}(A)$\\
           &&&$\upsilon_0(B)$\\
    \hline
  \end{tabular}\\ \ \\
  \caption{Distribution of local symmetries along weights}
  \label{tab:distr}
\end{table}
Obviously,
\begin{equation*}
  \abs{[X,Y]} = \abs{X} + \abs{Y}
\end{equation*}
for any homogeneous~$X$ and~$Y$.

From Equations~\eqref{eq:5a}, \eqref{eq:5b} and~\eqref{eq:6} we immediately
deduce the weights
\begin{equation*}
  \abs{q_i} = -i,\qquad \abs{r_i} = i + 2,\qquad i = 1, 2, \dots,
\end{equation*}
of the nonlocal variables of nonlocal variables in~$\tilde{\mathcal{E}}^+$
and~$\tilde{\mathcal{E}}^-$.

\section{Symmetries}
\label{sec:nonlocal-symmetries}

We describe here the Lie algebras~$\sym\tilde{\mathcal{E}}^+$
and~$\sym\tilde{\mathcal{E}}^-$.

\subsection{Symmetries in the positive hierarchy}
\label{sec:nonl-symmtr-posit}

Any symmetry of~$\tilde{\mathcal{E}}^+$ is a vector field
\begin{equation}\label{eq:28}
  \Ex_\Phi= \tilde{\Ev}_\phi + \sum_{i=1}^\infty \left(
    \phi_i\frac{\partial}{\partial q_i} + \sum_{j=1}^\infty
    \tilde{D}_y^j(\phi_i) \frac{\partial}{\partial q_i^{(j)}}\right),
\end{equation}
where~$\tilde{\Ev}_\phi$ is given by~\eqref{eq:7} with the total
derivatives~$\tilde{D}_\bullet$ instead of~$D_\bullet$ and the collection of
functions
\begin{equation*}
  \Phi = \langle \phi_0=\phi,\phi_1,\dots,
  \phi_i,\dots\rangle,
  \qquad \phi_0,\phi_i\in
  C^\infty(\tilde{\mathcal{E}}^+),
\end{equation*}
satisfies the equations
\begin{align}\label{eq:9}
  &\tilde{\ell}_{\mathcal{E}}(\phi) \equiv \tilde{D}_t\tilde{D}_y(\phi) -
    u_x\tilde{D}_x\tilde{D}_y(\phi) +
    u_y\tilde{D}_x^2(\phi) -
    u_{xy}\tilde{D}_x(\phi) +
    u_{xx}\tilde{D}_y(\phi) =
    0,\\\label{eq:10}
  &\tilde{D}_y(\phi)q_{1,t} + u_y\tilde{D}_t(\phi_1) =
    \tilde{D}_x(\phi),\\\label{eq:11}
  &\tilde{D}_y(\phi)q_{1,x} + u_y\tilde{D}_x(\phi_1)=0,\\\label{eq:12}
  &\tilde{D}_y(\phi)(q_{i,t} + q_{i-1,x}) + u_y(\tilde{D}_t(\phi_i) +
    \tilde{D}_x(\phi_{i-1})) = \tilde{D}_x(\phi)q_{i-1,y} +
    u_x\tilde{D}_x(\phi_{i-1}),\\ \label{eq:13}
  &\tilde{D}_y(\phi)q_{i,x} + u_y\tilde{D}_x(\phi_i) = \tilde{D}_y(\phi_{i-1}),
\end{align}
$i>1$. For any two symmetries~$\Phi$ and~$\Psi$ their \emph{Jacobi
  bracket}~$\{\Phi,\Psi\}$ is defined by
\begin{equation*}
  \Ex_{\{\Phi,\Psi\}} = [\Ex_\Phi,\Ex_\Psi].
\end{equation*}

\subsubsection{Lifts of local symmetries and hierarchies of nonlocal ones}
\label{sec:lifts-local-symm}

We begin with the following statement:
\begin{proposition}\label{sec:lifts-local-symm-prop-2}
  The local symmetries~$\psi_0$\textup{,} $\theta_{-2}(A)$\textup{,}
  $\theta_{-1}(A)$\textup{,} and~$\theta_0(A)$ can be lifted
  to~$\tilde{\mathcal{E}}^+$.
\end{proposition}

\begin{proof}
  Let us denote the desired lifts by
  \begin{align*}
    \Psi_0&=\langle \psi_0,\psi_0^1,\dots,\psi_0^i,\dots\rangle,\\
    \Theta_{-2}(A)&=\langle\theta_{-2}(A),\theta_{-2}^1(A), \dots,
                    \theta_{-2}^i(A),\dots\rangle,\\
    \Theta_{-1}(A)&=\langle\theta_{-1}(A),\theta_{-1}^1(A), \dots,
                    \theta_{-1}^i(A),\dots\rangle,\\
    \Theta_0(A)&=\langle\theta_0(A),\theta_0^1(A), \dots,
                    \theta_0^i(A),\dots\rangle
  \end{align*}
  and set
  \begin{align*}
    \psi_0^i &= iq_i +x q_{i,x},&& i\geq1,\\
    \theta_{-2}^i(A)&=0,&& i\geq1,\\
    \theta_{-1}^i(A)&=Aq_{i,x},&& i\geq1,\\
    \theta_0^1(A)&=\theta_{-1}(A)q_{1,x},\quad
                   \theta_0^i(A)=\theta_{-1}(A)q_{i,x} -Aq_{i-1,x},&& i>1.
  \end{align*}
  To establish that the above introduced functions are symmetries, we
  straightforwardly check that they satisfy
  Equations~\eqref{eq:10}--\eqref{eq:13}. For example, let us prove
  that~$\Psi_0$ is a symmetry.

  For Equation~\eqref{eq:10} one has\footnote{Here and below,
    the boxed terms cancel each other.}
  \begin{multline*}
    D_y(xu_x-2u)q_{1,x}+u_yD_x(q_1+xq_{1,x}) = (xu_{xy}\boxed{-2u_y})q_{1,x}
    +
    u_y(\boxed{2q_{1,x}}+xq_{1,xx}) \\
    =xu_{xy}q_{1,x} + u_yx\left(\frac{1}{u_y}\right)_x =
    xu_{xy}\frac{1}{u_y} -xu_y\frac{u_{xy}}{u_y^2}=0.
  \end{multline*}
  Now, Equation~\eqref{eq:11} reads
  \begin{multline*}
    D_y(xu_x-2u)q_{1,t} + u_yD_t(q_1+xq_{1,x}) - D_x(xu_x-2u)=\\
    (xu_{xy}\boxed{-2u_y})q_{1,t} +u_y(\boxed{q_{1,t}}+xq_{1,xt})+u_x-xu_{xx}
    = (xu_{xy} -
    u_y)\frac{u_x}{u_y} +xu_y\left(\frac{u_x}{u_y}\right)_x\\
    + u_x -xu_{xx}=(xu_{xy} - u_y)\frac{u_x}{u_y}
    +xu_y\frac{u_{xx}u_y-u_xu_{xy}}{u_y^2} + u_x=0.
  \end{multline*}
  Equation~\eqref{eq:12} acquires the form
  \begin{multline*}
    D_y(xu_x-2u)q_{i,x} + u_yD_x(iq_i+xq_{i,x}) -
    D_y((i-1)q_{i-1}+xq_{i-1,x})=\\
    (xu_{xy}\boxed{-2u_y})q_{i,x} + u_y(\boxed{(i+1)q_{i,x}}+xq_{i,xx}) -
    (i-1)q_{i-1,y}  -xq_{i-1,xy} =\\
    (xu_{xy}+\boxed{(i-1)u_y})\frac{q_{i-1,y}}{u_y} +
    u_yx\left(\frac{q_{i-1,y}}{u_y}\right)_x \boxed{-(i-1)q_{i-1,y}}=
    -xq_{i-1,xy}=\\
    xu_{xy}\frac{q_{i-1,y}}{u_y} +
    u_yx\frac{q_{i-1,xy}u_y-q_{i-1,y}u_{xy}}{u_y^2} -xq_{i-1,xy} = 0.
  \end{multline*}
  Finally, for Equation~\eqref{eq:13} one has
  \begin{multline*}
    D_y(xu_x-2u)\frac{u_x}{u_y}q_{i-1,y} + u_y(D_t(iq_i+xq_{i,x}) +
    D_x((i-1)q_{i-1}+xq_{i-1,x})) - \\D_x(xu_x-2u)q_{i-1,y} -
    u_xD_y((i-1)q_{i-1}+xq_{i-1,x}) =\\
    (xu_{xy}-2u_y)\frac{u_x}{u_y}q_{i-1,y} +
    (\boxed{iq_{i,t}}+xq_{i,xt}+\boxed{iq_{i-1,x}}+xq_{i-1xx}) +\\
    (u_x-xu_{xx})q_{i-1,y} - u_x((i-1)q_{i-1,y}+xq_{i-1,xy})=\\
    (xu_{xy}\boxed{-2u_y})\frac{u_x}{u_y}q_{i-1,y}
    +u_y\left(\boxed{i\frac{u_x}{u_y}q_{i-1,y}}
      +xq_{i,xt}+xq_{i-1,xx}\right)+\\
    (\boxed{u_x}-xu_{xx})q_{i-1,y} -
    u_x(\boxed{(i-1)q_{i-1,y}}+xq_{i-1,xy})=\\
    xu_{xy}\frac{u_x}{u_y}q_{i-1,y} +xu_y(\boxed{q_{i,t}+q_{i-1,x}})_x
    -xu_{xx}q_{i-1,y} - u_xxq_{i-1,xy}=\\
    xu_{xy}\frac{u_x}{u_y}q_{i-1,y}
    +xu_y\left(\frac{u_x}{u_y}q_{i-1,y}\right)_x -xu_{xx}q_{i-1,y} -
    u_xxq_{i-1,xy}=\\
    xu_{xy}\frac{u_x}{u_y}q_{i-1,y} +
    xu_y\left(\left(\frac{u_x}{u_y}\right)_xq_{i-1,y}
      +\boxed{\frac{u_x}{u_y}q_{i-1,xy}}\right) -xu_{xx}q_{i-1,y}
    \boxed{-u_xxq_{i,xy}}=\\
    xu_{xy}\frac{u_x}{u_y}q_{i-1,y}
    +xu_y\frac{u_{xx}u_y-u_xu_{xy}}{u_y^2}q_{i-1,y} -xu_{xx}q_{i-1,y},
  \end{multline*}
  and this finishes the proof.

  For other symmetries the proofs are similar.
\end{proof}

We shall now need a description of invisible symmetries
in~$\tilde{\mathcal{E}}^+$. We say that~$\Phi$ is an invisible symmetry of
\emph{depth~$k$} if its first~$k$ components vanish, i.e.,
\begin{equation*}
  \Phi=\langle \underbrace{0,\dots,0}_{k \text{ times}},\phi_1^{\mathrm{inv}},
  \dots, \phi_i^{\mathrm{inv}},\dots\rangle
\end{equation*}
The defining equations for invisible symmetries are
\begin{align*}
  &\tilde{D}_x(\phi_1^{\mathrm{inv}})=0,
  && \tilde{D}_t(\phi_1^{\mathrm{inv}})=0;\\
  &u_y\tilde{D}_x(\phi_i^{\mathrm{inv}}) =
    \tilde{D}_y(\phi_{i-1}^{\mathrm{inv}}),
  &&u_y(\tilde{D}_t(\phi_i^{\mathrm{inv}}) +
     \tilde{D}_x(\phi_{i-1}^{\mathrm{inv}})) =
     u_x\tilde{D}_x(\phi_{i-1}^{\mathrm{inv}}),\quad i>1.
\end{align*}
Then~$\phi_1^{\mathrm{inv}}=B(y)$ and any homogeneous symmetry of depth~$k$ is
completely determined by the function~$B$. Denote such a symmetry
by~$\Upsilon_k(B)$. One has
\begin{equation*}
  \abs{\Upsilon_k(B)} = k.
\end{equation*}

\begin{proposition}
  \label{sec:nonl-symm-posit-prop1}
  For any integer~$k\geq 1$ and a function~$B=B(y)$\textup{,} the
  symmetry~$\Upsilon_k(B)$ does exist.
\end{proposition}
\begin{proof}
  Consider the operator
  \begin{equation*}
    \mathcal{X} = q_1\frac{\partial}{\partial y} + \sum_{i=1}^\infty
    (i+1)q_{i+1}\frac{\partial}{\partial q_i}
  \end{equation*}
  and define
  \begin{equation}\label{eq:16}
    \phi_1^{\mathrm{inv}} = B(y), \qquad \phi_i^{\mathrm{inv}} =
    \frac{1}{i-1}\mathcal{X}(\phi_{i-1}^{\mathrm{inv}}), \quad i>1.
  \end{equation}
  Note that the defining equations for invisible symmetries can be rewritten
  in the form
  \begin{align*}
    &\frac{\partial \phi_2^{\mathrm{inv}}}{\partial q_1} = \frac{\partial
    B}{\partial y},\\
    &\dots\\
    &\frac{\partial\phi_i^{\mathrm{inv}}}{\partial q_{i-1}} =
      \frac{\partial\phi_{i-1}^{\mathrm{inv}}}{\partial q_{i-2}}, \dots,
      \frac{\partial \phi_i^{\mathrm{inv}}}{\partial q_1} = \frac{\partial
      \phi_{i-1}^{\mathrm{inv}}}{\partial y},\\
    &\dots
  \end{align*}
  Let us prove by induction the equalities
  \begin{equation*}
    \frac{\partial\phi_i^{\mathrm{inv}}}{\partial q_j} =
    \frac{\partial\phi_{i-1}^{\mathrm{inv}}}{\partial q_{j-1}}
  \end{equation*}
  (we formally set~$q_0=y$). The case~$i=2$ is checked by straightforward
  computations. Assume now that the statement is valid for some~$i>2$ and note
  that
  \begin{equation*}
    \left[\frac{\partial}{\partial q_j},\mathcal{X}\right] =
    j\frac{\partial}{\partial q_{j-1}}.
  \end{equation*}
  Then
  \begin{multline*}
    \frac{\partial\phi_{i+1}^{\mathrm{inv}}}{\partial q_j} = \frac{1}{i}
    \left( j\frac{\partial\phi_i^{\mathrm{inv}}}{\partial q_{j-1}} +
      \mathcal{X} \left(\frac{\partial \phi_i^{\mathrm{inv}}}{\partial
          q_j}\right)\right) = \frac{1}{i}\left(j\frac{\partial
        \phi_i^{\mathrm{inv}}}{\partial q_{j-1}} + \mathcal{X}
      \left(\frac{\partial \phi_{i-1}^{\mathrm{inv}}}{\partial
          q_{j-1}}\right)\right) \\
    = \frac{1}{i} \left( j \frac{\partial\phi_i^{\mathrm{inv}}}{\partial
        q_{j-1}} - (j-1)\frac{\partial\phi_{i-1}^{\mathrm{inv}}}{\partial
        q_{j-2}} + \frac{\partial
        \mathcal{X}(\phi_{i-1}^{\mathrm{inv}})}{\partial q_{j-1}}\right) =
    \frac{1}{i}\frac{\partial}{\partial q_{j-1}}\left(\phi_i^{\mathrm{inv}} +
      \mathcal{X}(\phi_{i-1}^{\mathrm{inv}})\right) \\
    = \frac{1}{i}\frac{\partial}{\partial q_{j-1}}(\phi_i^{\mathrm{inv}}
    +(i-1)\phi_i^{\mathrm{inv}}) = \frac{\partial
      \phi_i^{\mathrm{inv}}}{\partial q_{j-1}},
  \end{multline*}
  and this finishes the proof.
\end{proof}

Now, direct computations show that the functions
\begin{equation}\label{eq:14}
  \psi_{-1} = q_1u_y + x,\qquad
  \psi_{-2} = (2q_2 - q_1\,q_1^{(1)})u_y
\end{equation}
are shadows in the positive covering, i.e., they satisfy
Equation~\eqref{eq:9}.

\begin{proposition}
  \label{sec:lifts-local-symm-prop5}
  The shadows~\eqref{eq:14} can be extended to symmetries
  of~$\tilde{\mathcal{E}}^+$.
\end{proposition}
\begin{proof}
  Let us set
  \begin{equation*}
    \Psi_{-1}=\langle\psi_{-1},\psi_{-1}^1,\dots,\psi_{-1}^i,\dots\rangle,\qquad
    \Psi_{-2}=\langle\psi_{-2},\psi_{-2}^1,\dots,\psi_{-2}^i,\dots\rangle,
  \end{equation*}
  where
  \begin{equation*}
    \psi_{-1}^i = -(i+1)q_{i+1} + q_i^{(1)}q_1,\qquad
    \psi_{-2}^i = -(i+2)q_{i+2} + q_1q_{i+1}^{(1)} + (2q_2 - q_1q_1^{(1)})q_i^{(1)}.
  \end{equation*}
  The rest of the proof is similar to that of
  Proposition~\textup{\ref{sec:lifts-local-symm-prop-2}}
\end{proof}

Obviously,
\begin{equation*}
  \abs{\Psi_{-1}} = -1,\qquad\abs{\Psi_{-2}} = -2.
\end{equation*}

We now define two hierarchies of nonlocal symmetries by
\begin{align*}
  \Psi_{-k} &= \ad_{-1}^{k-2}(\Psi_{-2}),&&k\geq 3,\\
  \Upsilon_{-k}(B) &=\{\Psi_{-k-1},\Upsilon_1(B)\},&&k\geq0,
\end{align*}
where
\begin{equation*}
  \ad_{-1}(\Phi) = \{\Phi,\Psi_{-1}\}.
\end{equation*}
Obviously,
\begin{equation*}
  \abs{\Psi_{-k}} = \abs{\Upsilon_{-k}(B)} = -k
\end{equation*}
and~$\Upsilon_0(B)$ is an extension of the local symmetry~$\upsilon_0(B)$
to~$\tilde{\mathcal{E}}^+$. Elements of the
algebra~$\sym(\tilde{\mathcal{E}}^+)$ are distributed along weights as it is
indicated in Table~\ref{tab:nonloc}.
\begin{table}[bct]
  \centering
  \begin{tabular}{c|cccccccccc}\hline
    Weights:
    &\dots&$-l$&\dots&$-2$&$-1$&$0$&$1$&\dots&$k$&\dots\\
    \hline
    &\dots&$\Psi_{-l}$&\dots&$\Psi_{-2}$&$\Psi_{-1}$&$\Psi_0$&&&&\\
    &&&&$\Theta_{-2}(A)$&$\Theta_{-1}(A)$&$\Theta_0(A)$&&&&\\
    &\dots&$\Upsilon_{-l}(B)$&\dots&$\Upsilon_{-2}(B)$
                          &$\Upsilon_{-1}(B)$&$\Upsilon_0(B)$&$\Upsilon_1(B)$
                                       &\dots&$\Upsilon_k(B)$&\dots\\
    \hline
  \end{tabular}\\ \ \\
  \caption{Distribution of nonlocal symmetries in~$\tilde{\mathcal{E}}^+$
    along weights}
  \label{tab:nonloc}
\end{table}

\subsubsection{The Lie algebra structure}
\label{sec:lie-algebra-struct}
To compute the commutators, we shall need asymptotic estimates for coefficient
of symmetries that constitute a basis of~$\sym(\tilde{\mathcal{E}}^+)$.

We begin with the symmetries~$\Psi_{-k}$, $k\geq 1$, and we are interested in
the higher order terms (with respect to~$q_j$) of the coefficients
at~$\partial/\partial q_i$. Using the notation~\eqref{eq:28}, we have by
definition
\begin{align*}
  \Ex_{\Psi_{-1}} &= \dots +\left(-(i+1)q_{i+1} +
    q_1q_i^{(1)} + o(i-1)\right)\frac{\partial}{\partial q_i} + \dots,\\
  \Ex_{\Psi_{-2}} &= \dots \left(-(i+2)q_{i+2} + q_1q_{i+1}^{(1)} +
    o(i)\right)\frac{\partial}{\partial q_i} + \dots,
\end{align*}
where~$o(k)$ denotes the terms that contain~$q_j$ with~$j\leq k$.
 Assume
now that
\begin{equation*}
  \Ex_{\Psi_{-k}} = \dots + \left(a_k^iq_{i+k} +
    b_k^iq_1q_{i+k-1}^{(1)} + o(i+k-2)\right)\frac{\partial}{\partial q_i} +
  \dots
\end{equation*}
Then
\begin{align*}
  \Ex_{\Psi_{-k-1}} &= [\Ex_{\Psi_{-k}},\Ex_{\Psi_{-1}}] = \dots\\
              &+
                \left((i+k+1)a_k^i -
                (i+1)a_k^{i+1}q_{i+k+1}\phantom{q_{i+k}^{(1)}}\right.\\
              &\left.+ ((i+k)b_k^i-(i+1)b_k^{i+1})q_1q_{i+k}^{(1)} +
                o(i+k-1)\right)\frac{\partial}{\partial q_i} + \dots
\end{align*}
Thus
\begin{equation*}
  a_{k+1}^i = (i+k+1)a_k^i - (i+1)a_k^{i+1},\qquad
  b_{k+1}^i = (i+k)b_k^i -(i+1)b_k^{i+1}
\end{equation*}
and by elementary induction with the base~$a_2^i=-(i+2)$, $b_2^i=1$ we
immediately obtain
\begin{equation}
  \label{eq:15}
  a_k^i=-(k-2)!(k+i),\qquad b_k^i = (k-2)!
\end{equation}
for all~$i\geq 1$ (we formally set~$(-1)! = 1$). To comply with this result,
we change the basic element~$\Psi_0$ by~$\Psi_0\mapsto-\Psi_0$.

Now, we estimate the elements~$\Upsilon_k(B)$. For $k>0$ we use the
Definition~\eqref{eq:16} and by simple computations obtain that
\begin{equation*}
  \phi_i^{\mathrm{inv}} = B'q_{i-1} + B''q_1q_{i-2} + o(i-3)
\end{equation*}
and consequently
\begin{multline*}
  \Ex_{\Upsilon_k(B)}= \phi_1^{\mathrm{inv}}\frac{\partial}{\partial q_k} + \dots +
  \phi_{i-k+1}^{\mathrm{inv}}\frac{\partial}{\partial q_i} + \dots\\
  = B\frac{\partial}{\partial q_k}\dots +
  (B'q_{i-k} + B''q_1q_{i-k-1} + o(i-k-2))\frac{\partial}{\partial q_i} +
  \dots
\end{multline*}
Further,
\begin{multline*}
  \Ex_{\Upsilon_{-k}(B)} =  [\Ex_{\Psi_{-k-1}},\Ex_{\Upsilon_1(B)}] \\
  = \left[\dots + \left(a_{k+1}^iq_{i+k+1} + b_{k+1}^iq_1q_{i+k}^{(1)} +
      o(i+k-1)\right)\frac{\partial}{\partial q_i} +
    \dots,\right.\\
  \left.B\frac{\partial}{\partial q_1} + \dots + \left(B'q_{i-1} +
      B''q_1q_{i-2} + o(i-3)\right)\frac{\partial}{\partial q_i} +
    \dots\right] \\
  = \dots + \left((a_{k+1}^{i-1} - a_{k+1}^i)B'q_{i+k} -
    b_{k+1}^iBq_{i+k}^{(1)} +(a_{k+1}^{i-2} - a_{k+1}^i -
    b_{k+1}^i)B''q_1q_{i+k-1}\right.\\
  +\left. (b_{k+1}^{i-1} -
    b_{k+1}^i)B'q_1q_{i+k-1}^{(1)} + o(i+k-2)\right)\frac{\partial}{\partial
    q_i} + \dots \\
  = (k-1)!\left(\dots +\left(B'q_{i+k} -Bq_{i+k}^{(1)} +
      B''q_1q_{i+k-1}\right)\frac{\partial}{\partial q_i} + \dots\right)
\end{multline*}

Using the obtained estimates, we are ready to compute the commutators
now\footnote{Everywhere below we assume~$s!=1$ when~$s<0$.}:

\begin{proposition}
  \label{sec:nonl-symm-posit-prop2}
  One has the following commutator relations\textup{:}
  \begin{gather*}
    \{\Psi_{-k},\Psi_{-l}\} =
    \frac{(k-2)!(l-2)!}{(k+l-2)!}(k-l)\Psi_{-k-l},\quad k,l\geq 0,\\
    \{\Psi_{-k},\Upsilon_l(B)\} =
    \frac{l(-l-1)!(k-2)!}{(l-k-1)!}\Upsilon_{l-k}(B),\quad k\geq 0,\ l\in
    \mathbb{Z},\\\
    \{\Upsilon_k(B),\Upsilon_l(\tilde{B})\} =
    \frac{(-k-1)!(-l-1)!}{(-k-l-1)!}\Upsilon_{k+l}(B\tilde{B}' -
    B'\tilde{B}),\quad k,l\in\mathbb{Z}.
  \end{gather*}
\end{proposition}

\begin{proof}
  A neat use of the above deduced estimates.
\end{proof}

Let us change the initial basis by
\begin{equation*}
  \Psi_{-k}\mapsto \frac{1}{(k-2)!}\Psi_{-k}, \qquad
  \Upsilon_l(B)\mapsto \frac{1}{(-l-1)!}\Upsilon_l(B)
\end{equation*}
and recall a standard construction. Let~$\mathfrak{g}$ be a Lie
$\mathbb{R}$-algebra
and~$\mathbb{R}_n[z]=\mathbb{R}[z]/(z^n)$ be the ring of
truncated polynomials. Then the Lie
algebra~$\mathfrak{g}_{[n]}=\mathbb{R}_n[z]\otimes_{\mathbb{R}}\mathfrak{g}$
with the bracket
\begin{equation*}
  [a\otimes g,b\otimes h] = ab\otimes [g,h],\qquad g,h\in\mathfrak{g}\quad
  a,b\in\mathbb{R}_n[z],
\end{equation*}
is a graded Lie algebra
with~$\mathfrak{g}_0=\dots=\mathfrak{g}_{n-1}=\mathfrak{g}$ and all other
components being trivial. For polynomials in~$z^{-1}$ the similar construction
is denoted by~$\mathfrak{g}_{[-n]}$.  Denote also by~$\mathfrak{V}[t]$ the Lie
algebra of vector fields~$A(t)\partial/\partial t$ on~$\mathbb{R}$.  Then the
following result is valid:

\begin{theorem}
  \label{sec:lie-algebra-struct-theo1}
  The Lie algebra~$\sym(\tilde{\mathcal{E}}^+)$ is isomorphic to the
  semi-direct product of the non-positive part
  \begin{equation*}
    \mathfrak{W}^{-} = \left\{
        Z_k = z^{-k+1} \,\frac{\partial}{\partial z} \,\,\vert\,\, k \in
        \mathbb{N} \cup \{0\}\, 
      \right\}
    \end{equation*}
    of the Witt algebra with the direct sum
    $\mathfrak{L}[y]\oplus \mathfrak{V}_{[-3]}[t]$ of
    \begin{equation*}
      \mathfrak{L}[y] = \left\{
        Y_m(B) = z^{m} \,B(y)\,\frac{\partial}{\partial y} \,\,\vert\,\, m \in
        \mathbb{Z},\,\, 
        B \in C^{\infty}(\mathbb{R})\,
      \right\}
    \end{equation*}
    and
    \begin{equation*}
      \mathfrak{V}[t]_{[-3]} = \left\{
        X_s(A) = z^{s} \,A(t)\,\frac{\partial}{\partial t} \,\,\vert\,\, s \in
        \{0,1,2\},\,\, 
        A \in C^{\infty}(\mathbb{R})\,
      \right\}
    \end{equation*}
    with the natural action of~$z^{-k+1}\partial/\partial z$
    on~$\mathfrak{L}[y]$ and~$\mathfrak{V}[t]_{[-3]}$.
\end{theorem}

In the theorem above the isomorphism maps $\Psi_{-k}$ to $Z_k$,
$\Upsilon_m(B)$ to
~$Y_m(B)$ and
$\Theta_{-s}(A)$ to $X_s(A)$.

\subsection{Symmetries in the negative hierarchy}
\label{sec:nonl-symmtr-negat}

Using Proposition~\ref{sec:negative-hierarchy-gauge}, we set~$F=1$ in the
defining equations of the negative hierarchy. After such a simplification, the
study of the negative case becomes quite similar to that of the positive
one. Any symmetry in~$\tilde{\mathcal{E}}^-$ is a vector field
\begin{equation}\label{eq:29}
  \Ex_\phi =\tilde{\Ev}_\phi +\sum_{i=1}^\infty
  \left(\phi_i\frac{\partial}{\partial r_i} + \sum_{j=1}^\infty
    \tilde{D}_t^j(\phi_i)\frac{\partial}{\partial r_i^{(j)}}\right),
\end{equation}
where~$\tilde{\Ev}_\phi$ with the total derivatives on~$\tilde{\mathcal{E}}^-$
and
\begin{equation*}
  \Phi = \langle\phi_0=\phi,\phi_1,\dots,\phi_i,\dots\rangle,\qquad \phi_i\in
  C^\infty(\tilde{\mathcal{E}}^-),
\end{equation*}
satisfies the equations
\begin{align}\label{eq:19}
   &\tilde{\ell}_{\mathcal{E}}(\phi) \equiv \tilde{D}_t\tilde{D}_y(\phi) -
    u_x\tilde{D}_x\tilde{D}_y(\phi) +
    u_y\tilde{D}_x^2(\phi) -
    u_{xy}\tilde{D}_x(\phi) +
    u_{xx}\tilde{D}_y(\phi) =
     0,\\ \label{eq:20}
  &\tilde{D}_x(\phi_1) = \tilde{D}_t(\phi) -
    2u_x\tilde{D}_x(\phi),\\ \label{eq:21}
  &\tilde{D}_y(\phi_1) = -u_y\tilde{D}_x(\phi) -
    u_x\tilde{D}_y(\phi),\\ \label{eq:22}
  &\tilde{D}_x(\phi_i) = r_{i-1,x}\tilde{D}_x(\phi) +
    u_x\tilde{D}_x(\phi_{i-1}) - \tilde{D}_t(\phi_{i-1}),\\ \label{eq:23}
  &\tilde{D}_y(\phi_i) = r_{i-1,x}\tilde{D}_y(\phi) + u_y\tilde{D}_x(\phi_{i-1}),
\end{align}
$i>1$. Like in Section~\ref{sec:nonl-symmtr-posit}, for any two
symmetries~$\Phi$ and~$\Psi$ their Jacobi bracket~$\{\Phi,\Psi\}$ is defined
by
\begin{equation*}
  \Ex_{\{\Phi,\Psi\}} = [\Ex_\Phi,\Ex_\psi].
\end{equation*}

\subsubsection{Lifts of local symmetries and hierarchies of nonlocal ones}
\label{sec:lifts-local-symm-1}

In what follows, we shall need the operator
\begin{equation}
  \label{eq:18}
  \mathcal{Y}_+ = -x\frac{\partial}{\partial t} + 2u\frac{\partial}{\partial x}
  + 3r_1\frac{\partial}{\partial u} + \sum_{i\geq 1}
  (i+3)r_{i+1}\frac{\partial}{\partial r_i}
\end{equation}

\begin{proposition}
  \label{sec:nonl-symm-negat-prop3}
  The symmetries~$\psi_0$\textup{,} $\upsilon(B)$\textup{,}
  and~$\theta_{-2}(A)$ can be lifted to~$\tilde{\mathcal{E}}^-$.
\end{proposition}

\begin{proof}
  We denote the lifts by
  \begin{align*}
    \Psi_0&=\langle\psi_0,\psi_0^1,\dots,\psi_0^i,\dots\rangle,\\
    \Upsilon_0(B)&=\langle\upsilon_0(B),\upsilon_0^1(B), \dots,
                   \upsilon_0^i(B), \dots\rangle,\\
    \Theta_{-2}(A)&=\langle\theta_{-2}(A), \theta_{-2}^1(A), \dots,
                    \theta_{-2}^i(A), \dots\rangle
  \end{align*}
  and set
  \begin{align*}
    \psi_0^i&=-(i+2)r_i+xr_{i,x},&&i\geq 1,\\
    \upsilon_0^i(B)&=Br_{i,y},&&i\geq 1,\\
    \theta_{-2}^1(A)&=-xA',\quad
                      \theta_{-2}^i(A)
                      =\frac{1}{i}\mathcal{Y}_{+}(\theta_{-2}^{i-1}(A))&&i>1.
  \end{align*}
  The rest of the proof is similar to that of
  Proposition~\ref{sec:lifts-local-symm-prop-2}.
\end{proof}

The next step is to describe invisible symmetries. These symmetries must
satisfy
\begin{align*}
  &\tilde{D}_x(\phi_1)=0,
  &&\tilde{D}_y(\phi_1)=0,\\
  &\tilde{D}_x(\phi_i)=u_x\tilde{D}_x(\phi_{i-1}) - \tilde{D}_t(\phi_{i-1}),
  &&\tilde{D}_y(\phi_i)=u_y\tilde{D}_x(\phi_{i-1}),
\end{align*}
where~$i>1$.

\begin{proposition}
  For every~$A=A(t)$ and~$k\geq3$ there exists a unique invisible
  symmetry~$\Theta_{-k}(A)$ of weight~$\abs{\Theta_{-k}(A)}=-k$.
\end{proposition}

\begin{proof}
  Let us use the notation~\eqref{eq:29} and set
  \begin{equation*}
    \Ex_{\Theta_{-k}(A)}=\phi_1^{\mathrm{inv}}\frac{\partial}{\partial r_{k-2}} +
    \dots + \phi_i^{\mathrm{inv}}\frac{\partial}{\partial r_{k+i-3}} + \dots,
  \end{equation*}
  where~$\phi_1^{\mathrm{inv}}=A$
  \begin{equation*}
    \phi_i^{\mathrm{inv}}
    =\frac{1}{i-1}\mathcal{Y}_{+}(\phi_{i-1}^{\mathrm{inv}}), \qquad i>1.
  \end{equation*}
  The proof is accomplished by induction on~$i$.
\end{proof}

Consider now two functions
\begin{equation*}
  \psi_1=3r_1+xu_t-2uu_x,\qquad \psi_2=4r_2+xr_1^{(1)}+2uu_t-(xu_t+3r_1)u_x.
\end{equation*}
It is straightforwardly checked that they are shadows
in~$\tilde{\mathcal{E}}^-$, i.e., satisfy Equation~\eqref{eq:19}.

\begin{proposition}
  The shadows~$\psi_1$ and~$\psi_2$ are extended to nonlocal symmetries
  of~$\tilde{\mathcal{E}}^-$.
\end{proposition}

\begin{proof}
  It suffices to set
  \begin{equation*}
    \Psi_1=\langle\psi_1,\psi_1^1,\dots,\psi_1^i,\dots\rangle
  \end{equation*}
  with
  \begin{equation*}
    \psi_1^i=(i+3)r_{i+1} + xr_i^{(1)} - 2ur_{i,x}
  \end{equation*}
  and
  \begin{equation*}
    \Psi_2=\langle\psi_2,\psi_2^2,\dots,\psi_2^i,\dots\rangle
  \end{equation*}
  with
  \begin{equation*}
    \psi_2^i = (i+4)r_{i+2} + xr_{i+1}^{(1)} + 2ur_{i}^{(1)} - (xu_t + 3r_1)r_{i,x}.
  \end{equation*}
  The rest of the proof is a straightforward check of
  Equations~\eqref{eq:20}--\eqref{eq:23}.
\end{proof}
Obviously,
\begin{equation*}
  \abs{\Psi_1} = 1, \qquad \abs{\Psi_2} = 2.
\end{equation*}

Similar to the positive case, we define now the first hierarchy of nonlocal
symmetries by setting
\begin{equation*}
  \Psi_k=\ad_{+1}^{k-2}(\Psi_2),\qquad k\geq 3,
\end{equation*}
where~$\ad_{+1}(\Phi)=\{\Psi_1,\Phi\}$. One has,
\begin{equation*}
  \abs{\Psi_k} = k.
\end{equation*}
The second hierarchy will be defined in the next subsection.

\subsubsection{The Lie algebra structure}
\label{sec:lie-algebra-struct-1}

As above, we need asymptotic estimates to compute the commutators. Similar to
the positive case, we establish by induction the following estimates for the
symmetries~$\Psi_k$:
\begin{equation*}
  \Ex_{\Psi_k}= \dots+\left(a_k^ir_{i+k} + b_k^ixr_{i+k-1}^{(1)} +
  o(i+k-2)\right)\frac{\partial}{\partial r_i} + \dots,
\end{equation*}
where
\begin{equation*}
  a_k^i = (k-2)!(i+k+2),\qquad b_k^i=(k-2)!
\end{equation*}
To have the unified signs, we also rescale~$\Psi_0\mapsto-\Psi_0$.  Using this
estimate, we easily prove the following
\begin{proposition}
  One has the following commutator relations
  \begin{equation*}
    \{\Psi_k,\Psi_l\} = \frac{(l-2)!(k-2)!(l-k)}{(k+l-2)!}\Psi_{k+l}
  \end{equation*}
  for all~$k$\textup{,} $l\geq 0$.
\end{proposition}
Of course, its is natural to rescale the elements~$\Psi_k$
by~$\Psi_k\mapsto\Psi_k/(k-2)!$ and obtain the commutators
\begin{equation*}
  \{\Psi_k,\Psi_l\} = (l-k)\Psi_{k+l}.
\end{equation*}
Thus, for the new~$\Psi_k$ the estimate is
\begin{equation*}
  \Ex_{\Psi_k}= \dots+\left((i+k+2)r_{i+k} + xr_{i+k-1}^{(1)} +
  o(i+k-2)\right)\frac{\partial}{\partial r_i} + \dots,
\end{equation*}

We now complete the sequence of symmetries~$\{\Theta_{-k}(A)\}$, $k\leq 0$, by
setting
\begin{equation}\label{eq:17}
  \Theta_k(A)=-\frac{1}{3}\{\Psi_{k+3},\Theta_{-3}(A)\},\qquad k\geq -2.
\end{equation}
One has~$\abs{\Theta_k(A)} = k$, and elements of~$\sym(\tilde{\mathcal{E}}^-)$
are distributed along the weights as indicated in Table~\ref{tab:distr-}.
\begin{table}[bct]
  \centering
  \begin{tabular}{c|cccccccccc}\hline
    Weights:
    &\dots&$-l$&\dots&$-2$&$-1$&$0$&$1$&\dots&$k$&\dots\\
    \hline
    &&&&&&$\Psi_0$&$\Psi_1$&\dots&$\Psi_k$&\dots\\
    &\dots
          &$\Theta_{-l}(A)$&\dots&$\Theta_{-2}(A)$
                          &$\Theta_{-1}(A)$&$\Theta_0(A)$&
                                                           $\Theta_1(A)$&\dots&$\Theta_k(A)$&\dots\\
    &&&&&&$\Upsilon_0(B)$&&&&\\
    \hline
  \end{tabular}
  \caption{Distribution of~$\sym(\tilde{\mathcal{E}}^-)$ along the weights}
  \label{tab:distr-}
\end{table}

The coefficients of invisible symmetries are
\begin{align*}
  \phi_1^{\mathrm{inv}} &= A,\\
  \phi_2^{\mathrm{inv}} &= -xA',\\
  \phi_3^{\mathrm{inv}} &= -uA + \frac{1}{2}x^2A'',\\
  \phi_4^{\mathrm{inv}} &= -r_1A' +uxA'' -\frac{1}{6}x^3A''',
\end{align*}
while for~$i\geq 5$ we have the estimates
\begin{equation*}
  \phi_i^{\mathrm{inv}} = -A'r_{i-3} +xA''r_{i-4} + o(i-5).
\end{equation*}
Thus,
\begin{multline*}
  \Ex_{\Theta_{-k}(A)} = A\frac{\partial}{\partial r_{k-2}} + \dots +
  \phi_{i-k+3}^{\mathrm{inv}}\frac{\partial}{\partial r_i} + \dots\\
  = A\frac{\partial}{\partial r_{k-2}} + \dots + \left(-A'r_{i-k} +xA''r_{i-k-1}
    + o(i-k-2)\right)\frac{\partial}{\partial r_i} + \dots
\end{multline*}
for~$k\geq 3$.

Using the obtained estimates for~$\Psi_k$ and~$\Theta_{-3}(A)$, we get
\begin{multline*}
  \Ex_{\Theta_k(A)} = -\frac{1}{3}[\Ex_{\Psi_{k+3}},\Ex_{\Theta_{-3}(A)}] \\=
  \left[\dots + 
    \left((i+k+2)r_{i+k} + xr_{i+k-1}^{(1)} +
      o(i+k-2)\right)\frac{\partial}{\partial r_i}+ \dots\right.,\\
  \left.\dots + \left(-A'r_{i-3} +xA''r_{i-4} +
      o(i-5)\right)\frac{\partial}{\partial r_i}+ \dots\right]\\
  =\dots+(-A'r_{i+k}+xA''r_{i+k-1} + o(i+k-2))\frac{\partial}{\partial r_i} +\dots
\end{multline*}
for all~$k\geq -2$. These estimates lead directly to
\begin{proposition}
  One has
  \begin{equation*}
    \{\Psi_k,\Theta_l(A)\} = l\,\Theta_{k+l}(A)
  \end{equation*}
  for all~$k\geq 0$\textup{,} $l\in\mathbb{Z}$.
\end{proposition}

Finally, we have
\begin{proposition}
  \label{sec:nonl-symm-negat-prop4}
  One has
  \begin{equation*}
    \{\Theta_k(A),\Theta_l(\tilde{A})\} =\Theta_{k+l}(A\tilde{A}' - A'\tilde{A})
  \end{equation*}
  for all~$k$\textup{,} $l\in\mathbb{Z}$ and smooth
  functions~$A=A(t)$\textup{,} $\tilde{A}=\tilde{A}(t)$.
\end{proposition}
\begin{proof}
  The result easily follows from the above estimates when~$k$ or~$l\leq-3$,
  but the method does not work when both~$k$ and~$l>-3$. Nevertheless, one has
  in this case
  \begin{multline*}
    \{\Theta_k(A),\Theta_l(\tilde{A})\} =
    -\frac{1}{3}\{\{
    \Psi_{k+3},\Theta_{-3}(A)\},\Theta_l(\tilde{A})\}
    =
    -\frac{1}{3}(\{\{\Psi_{k+3},\Theta_l(\tilde{A})\},\Theta_{-3}(A)\}\\
    + \{\Psi_{k+3},\{\Theta_{-3}(A),\Theta_l(\tilde{A})\}\}) =
    -\frac{1}{3}(l\,\{\Theta_{k+l+3}(\tilde{A}),\Theta_{-3}(A)\} +
    \{\Psi_{k+3},\Theta_{l-3}(A\tilde{A}' -    A'\tilde{A})\})\\
    = -\frac{1}{3}\,(-l\,\Theta_{k+l}(A\tilde{A}' - A'\tilde{A}) +
    (l-3)\,\Theta_{k+l}(A\tilde{A}' - A'\tilde{A})) = \Theta_{k+l}(A\tilde{A}'
    -    A'\tilde{A}), 
  \end{multline*}
  and this finishes the proof.
\end{proof}

Thus we have the result similar to Theorem~\ref{sec:lie-algebra-struct-theo1}:

\begin{theorem}\label{sec:lie-algebra-struct-theo-2}
  The Lie algebra~$\sym(\tilde{\mathcal{E}}^-)$ is isomorphic to the direct
  sum
  $\left(\mathfrak{W}^{+} \ltimes \mathfrak{L}[t]\right) \oplus
  \mathfrak{V}[y]$ of the semi-direct product of the positive part
  \begin{equation*}
    \mathfrak{W}^{+} = \left\{
        z^{k+1} \,\frac{\partial}{\partial z} \,\,\vert\,\, k \in \mathbb{N}
        \cup \{0\}\, 
      \right\}
    \end{equation*}
    of the Witt algebra with
    \begin{equation*}
      \mathfrak{L}[t] = \left\{
        z^{m} \,A(t)\,\frac{\partial}{\partial t} \,\,\vert\,\, m \in
        \mathbb{Z},\,\, 
        A \in C^{\infty}(\mathbb{R})\,
      \right\},
    \end{equation*}
    where the vector fields $z^{k+1}\partial/\partial z$ act naturally on
    $\mathfrak{L}[t]$\textup{,} and
    \begin{equation*}
      \mathfrak{V}[y] = \left\{
        B(y)\,\frac{\partial}{\partial y} \,\,\vert\,\, B \in
        C^{\infty}(\mathbb{R})\, 
      \right\}
    \end{equation*}
    is the Lie algebra of vector fields on the line.
\end{theorem}

\section{Action of the recursion operators}
\label{sec:acti-recurs-oper}

We discuss here the action of recursion operators in the hierarchies of
nonlocal symmetries described above.

\subsection{Action of the recursion operator to local symmetries and shadows}

The algebra~$\sym(\mathcal{E})$ admits a recursion
operator~$\hat{\chi} = \mathcal{R}_{+}(\chi)$ defined by the system
\begin{equation}\label{recursion_operator}
  \begin{array}{rcl}
    D_t (\hat{\chi}) &=&
                         u_y^{-1}\,\big(u_y\,D_x(\chi)-
                         u_x\,D_y(\chi)+(u_xu_{xy}-u_yu_{xx}) \hat{\chi}\big),
                             \\[2mm]
    D_x (\hat{\chi})&=& u_y^{-1}\,\big(u_{xy}\,\hat{\chi}-D_y(\chi)\big),
  \end{array}
\end{equation}
see~\cite{Morozov2012}.  This means that~$\hat{\chi}$ is a solution
to~\eqref{linearized_eq} whenever~$\chi$ is.  Another recursion
operator~$\chi= \mathcal{R}_{-}(\hat{\chi})$ is given by the system
\begin{equation} \label{inverse_recursion_operator}
  \begin{array}{rcl}
    D_x (\chi) &=& D_t(\hat{\chi}) - u_x\,D_x(\hat{\chi}) +u_{xx}\,\hat{\chi},
    \\[2mm]
    D_y (\chi) &=& - u_y\,D_x(\hat{\chi})+ u_{xy}\,\hat{\chi}.
  \end{array}
\end{equation}
The operators~$\mathcal{R}_+$ and~$\mathcal{R}_-$ are mutually inverse.

The actions of $\mathcal{R}_{+}$ and $\mathcal{R}_{-}$ on $\sym(\mathcal{E})$
may be prolonged to the shadows of nonlocal symmetries from
$\sym(\tilde{\mathcal{E}}^+)$ and $\sym(\tilde{\mathcal{E}}^-)$ if we replace
the derivatives $D_t$, $D_x$ and $D_y$ in~\eqref{recursion_operator}
and~\eqref{inverse_recursion_operator} by $\hat{D}_t$, $\hat{D}_x$ and
$\hat{D}_y$ defined as
\begin{align*}
  \hat{D}_x& = D_x +
             \sum_{j=0}^\infty \hat{D}_y^j
             \left(\frac{1}{u_y}\right)\frac{\partial}{\partial q_1^{(j)}} 
             + \sum_{i=2}^\infty\sum_{j=0}^\infty
             \hat{D}_y^j
             \left(\frac{q_{i-1}^{(1)}}{u_y}\right)\frac{\partial}{\partial 
             q_i^{(j)}} \\
           &+ \sum_{j=0}^\infty \hat{D}_t^j(u_x^2 -
             u_t)\frac{\partial}{\partial r_1^{(j)}} 
             + \sum_{i=2}^\infty \sum_{j=0}^\infty \hat{D}_t^j(u_xr_{i-1,x}-
             r_{i-1,t})\frac{\partial}{\partial r_i^{(j)}}, 
  \\
  \hat{D}_y&=D_y
             + \sum_{i=1}^\infty\sum_{j=0}^\infty
             q_i^{(j+1)}\frac{\partial}{\partial q_i^{(j)}} 
             + \sum_{j=0}^\infty \hat{D}_t^j(u_x
             \,u_y)\frac{\partial}{\partial r_1^{(j)}} 
             + \sum_{i=2}^\infty \sum_{j=0}^\infty
             \hat{D}_t^j(u_yr_{i-1,x})\frac{\partial}{\partial r_i^{(j)}}, 
  \\
  \hat{D}_t&=D_t
             + \sum_{j=0}^\infty
             \hat{D}_y^j\left(\frac{u_x}{u_y}\right)\frac{\partial}{\partial
             q_1^{(j)}} 
             + \sum_{i=2}^\infty\sum_{j=0}^\infty
             \hat{D}_y^j\left(\frac{u_x}{u_y}q_{i-1}^{(1)} - \hat{D}_x\left(
             q_{i-1}^{(0)}\right)\right)\frac{\partial}{\partial q_i^{(j)}} 
  \\
           &+ \sum_{i=1}^\infty \sum_{j=0}^\infty r_i^{(j+1)}
             \frac{\partial}{\partial r_i^{(j)}},
\end{align*}
that is, consider the Whitney product of the coverings
$\tilde{\mathcal{E}}^{+}$ and $\tilde{\mathcal{E}}^{-}$. The results of the
replacement will be also denoted by $\mathcal{R}_{+}$ and $\mathcal{R}_{-}$.

Note that the operators act nontrivially on `vacuum':
\begin{equation*}
  \mathcal{R}_+(0)=\theta_{-2}(A),\qquad \mathcal{R}_-(0)=\upsilon_0(B),
\end{equation*}
which immediately follows from Equations~\eqref{recursion_operator}
and~\eqref{inverse_recursion_operator}; thus the actions are reasonable to
consider modulo~$\theta_{-2}(A)$ for~$\mathcal{R}_+$ and~$\upsilon_0(B)$
for~$\mathcal{R}_-$. Taking into account this remark, we have the following

\begin{proposition}
  Modulo images of the trivial symmetry\textup{,} the action of recursion
  operators is of the form
  \begin{align*}
    &\mathcal{R}_+(\theta_i(A))=\begin{cases}
      \alpha_i^+\theta_{i-1}(A),
      &i>-2,\\
      0,&i=-2,
    \end{cases}
      &&\mathcal{R}_-(\theta_i(A))=\alpha_i^-\theta_{i+1}(A),\quad
         i\geq-2,
    \\
    &\mathcal{R}_+(\upsilon_i(B))=\beta_i^+\upsilon_{i+1}(B),\quad i\leq 0,
      &&\mathcal{R}_-(\upsilon_i(B))=
         \begin{cases}
           \beta_i^-\upsilon_{i+1}(B),&i<0,\\
           0,&i=0,
         \end{cases}
    \\
    &\mathcal{R}_+(\psi_i)=\gamma_i^+\psi_{i-1},
      &&\mathcal{R}_-(\psi_i)=\gamma_i^-\psi_{i+1},
         \qquad i\in \mathbb{Z},  
  \end{align*}
  where~$\alpha_i^\pm$\textup{,} $\beta_i^\pm$\textup{,} and~$\gamma_i^\pm$
  are nonzero constants.
\end{proposition}

\begin{proof}
  It suffices to notice that the weights of~$\mathcal{R}_+$
  and~$\mathcal{R}_-$ are~$-1$ and~$+1$, respectively, that their action
  (Modulo images of~$0$) does not change the dependence of shadows on~$y$
  and~$t$, and that the only shadows that may be taken to~$0$
  are~$\theta_{-2}(A)$ and~$\upsilon_0(B)$.
\end{proof}

Note that the recursion operators $\mathcal{R}_{+}$ and $\mathcal{R}_{-}$
`glue together' the shadows $\psi_m$ of nonlocal symmetries in coverings
$\tilde{\mathcal{E}}^{+}$ and $\tilde{\mathcal{E}}^{-}$ and `tunnel' from the
series of $\theta_k(A)$ to that of $\upsilon_{j}(B)$, see
Table~\ref{tab:RO-act}.

\begin{table}[bct]
  \centering
  \begin{equation*}
    \xymatrix{
      &\dots\ar@<1ex>[r]^{\mathcal{R}_-}
      &\ar[l]^{\mathcal{R}_+}\psi_{-1}\ar@<1ex>[r]^{\mathcal{R}_-}
      &\ar[l]^{\mathcal{R}_+}\psi_0\ar@<1ex>[r]^{\mathcal{R}_-}
      &\ar[l]^{\mathcal{R}_+}\psi_1\ar@<1ex>[r]^{\mathcal{R}_-}
      &\ar[l]^{\mathcal{R}_+}\dots&\\
      \dots\ar@<1ex>[r]^{\mathcal{R}_-}
      &\ar[l]^{\mathcal{R}_+}\upsilon_{-1}(B)\ar@<1ex>[r]^{\mathcal{R}_-}
      &\ar[l]^{\mathcal{R}_+}\upsilon_0(B)\ar@<1ex>[r]^{\mathcal{R}_-}
      &\ar[l]^{\mathcal{R}_+}0\ar@<1ex>[r]^{\mathcal{R}_-}
      &\ar[l]^{\mathcal{R}_+}\theta_{-2}(A)\ar@<1ex>[r]^{\mathcal{R}_-}
      &\ar[l]^{\mathcal{R}_+}\theta_{-1}(A)\ar@<1ex>[r]^{\mathcal{R}_-}
      &\ar[l]^{\mathcal{R}_+}\dots
    }
  \end{equation*}
  \caption{Action of recursion operators}
  \label{tab:RO-act}
\end{table}

\subsection{Recursion relations for symmetries of the positive covering}
In this section we find an operator that provides an alternative way to
construct elements of~$\sym(\tilde{\mathcal{E}}^{+})$. To this end, we express
$u_x$, $u_y$ from~\eqref{eq:5a}:
\begin{equation}
  \label{eq:5c}
  \begin{array}{lcl}
    u_{x}&=&\dfrac{q_{1,t}}{q_{1,x}},\\[3mm]
    u_{y}&=&\dfrac{1}{q_{1,x}}.
  \end{array}\qquad
\end{equation}
This system is compatible whenever equation
\begin{equation}
  q_{1,xx} = q_{1,t}q_{1,xy}-q_{1,x}q_{1,ty}.
  \label{uhe_eq}
\end{equation}
holds. This equation is known as the universal hierarchy equation,
see~\cite{MartinezAlonsoShabat2002,MartinezAlonsoShabat2004}.  Thus
systems~\eqref{eq:5a} and~\eqref{eq:5c} define a B\"acklund transformation
between~\eqref{eq:1} and~\eqref{uhe_eq}, see~\cite{Morozov2012_SIGMA}.
Substituting~\eqref{eq:5c} to~\eqref{eq:5b} yields
\begin{equation}
  \label{eq:5d}
  \begin{array}{lcl}
    q_{k,t}&=&q_{1,t} q_{k-1,y} - q_{k-1,x},\\[2mm]
    q_{k,x}&=&q_{1,x} q_{k-1,y},
  \end{array}\qquad
\end{equation}
where $k \ge 2$. The compatibility conditions for this system after renaming
$k-1 \mapsto k$ get the form
\begin{equation}
  q_{k,xx} = q_{1,t}q_{k,xy}-q_{1,x}q_{k,ty}, \qquad k \ge 2.
  \label{extra_uhe_eq}
\end{equation}

\begin{proposition}
  Systems~\eqref{eq:5c} and~\eqref{eq:5d} define a B\"acklund
  auto-transformation for the infinite system of {\sc
    pde}s~\eqref{uhe_eq}\textup{,} \eqref{extra_uhe_eq}.
\end{proposition}

\begin{proof}
  The compatibility conditions of~\eqref{eq:5c}, \eqref{eq:5d} are definitions
  for~\eqref{uhe_eq} and~\eqref{extra_uhe_eq}. From~\eqref{eq:5d} we have the
  inverse transformation
  \begin{equation*}
    \begin{array}{lcl}
      q_{k-1,x}&=&-q_{k,t} +\dfrac{q_{1,t}}{q_{1,x}} \, q_{k,x},\\[3mm]
      q_{k-1,y}&=&\dfrac{q_{k,x}}{q_{1,x}},
    \end{array}\qquad
  \end{equation*}
  whose compatibility conditions also coincide with~\eqref{extra_uhe_eq}.
\end{proof}

\begin{corollary}
  The linearization
  \begin{eqnarray}
    D_t(\hat{\chi}_1)
    &=& u_y^{-2} \,\left(u_y \,D_x(\chi_0) - u_x \,D_y(\chi_0)\right),
        \label{lifted_recursion_operator_1}
    \\
    D_x(\hat{\chi}_1) &=& -u_y^{-2} \,D_y(\chi_0),
                          \label{lifted_recursion_operator_2}
    \\
    D_t(\hat{\chi}_k) &=& q_{1,t} D_y(\chi_{k-1})+q_{k-1,y} D_t(\chi_1)-
                          D_x(\chi_{k-1}),
                          \label{lifted_recursion_operator_3}
    \\
    D_x(\hat{\chi}_k) &=& q_{1,x} D_y(\chi_{k-1})+q_{k-1,y} D_x(\chi_1)
                          \label{lifted_recursion_operator_4}
  \end{eqnarray}
  of~\eqref{eq:5a} and~\eqref{eq:5d} defines a recursion operator
  \begin{equation*}
    \mathcal{Q}((\chi_0,\chi_1,\chi_2, \dots , \chi_k, \dots))
    =(\chi_0,\hat{\chi}_1,\hat{\chi}_2, \dots , \hat{\chi}_k, \dots)
  \end{equation*}
  for $\sym(\tilde{\mathcal{E}}^{+})$.
\end{corollary}

Note that symmetries~$\xi$ and~$\mathcal{Q}(\xi)$ have the same shadows and
consequently differ by an invisible symmetry. Thus, at first glance, the
recursion operator~$\mathcal{Q}$ seems to be useless, but this is not the
case: it provides an alternative way for lifting shadows to nonlocal
symmetries in $\tilde{\mathcal{E}}^{+}$. Namely, take a local symmetry or a
shadow $\chi_0$, then~\eqref{lifted_recursion_operator_1},
\eqref{lifted_recursion_operator_2} give $\chi_1$,
applying~\eqref{lifted_recursion_operator_3},
\eqref{lifted_recursion_operator_4} with $k=2$ to $\chi_1$ gives $\chi_2$,
etc., applying~\eqref{lifted_recursion_operator_3},
\eqref{lifted_recursion_operator_4} with $k=m$ to $\chi_{m-1}$ gives $\chi_m$,
etc.

\begin{proposition}
  The following relations are valid\textup{:}
  \begin{align*}
    &\mathcal{Q}(\psi_0) = \psi_0^1,
    &&\mathcal{Q}(\psi^k_0) = \psi_0^{k+1},
    \\
    &\mathcal{Q}(\theta_{-j}(A)) = \theta_{-j}^1(A),
    &&
       \mathcal{Q}(\theta_{-j}^k(A)) = \theta_{-j}^{k+1}(A)
  \end{align*}
  for all $k \ge 1$ and~$j =0$, $1$, $2$. One also has
  \begin{equation*}
    \mathcal{Q}(\upsilon_0(B)) = \upsilon_0^1(B),
    \qquad
    \mathcal{Q}(\upsilon_0^k(B)) = \upsilon_0^{k+1}(B).
  \end{equation*}
\end{proposition}

The proof is fulfilled along the same lines as that of
Proposition~\ref{sec:lifts-local-symm-prop-2}.

Unfortunately, we could not find a similar recursion operator for nonlocal
symmetries in the negative covering.

\section{Conclusion}
\label{sec:discussion}

We gave a complete description of nonlocal symmetries associated to the Lax
representation of the 3D rdDym equation. The revealed Lie algebra structure of
these symmetries seems quite interesting and we intend to study nonlocal
symmetries of other Lax integrable equations from~\cite{Fer-Mos} in the
forthcoming research.

\section*{Acknowledgments}
\label{sec:acknowledgements}
Computations were supported by the \textsc{Jets} software,~\cite{Jets}. The
second author (ISK) was partially supported by the Simons-IUM fellowship.

\end{document}